\documentclass[10pt,letterpaper]{article}
\usepackage[ruled]{algorithm2e}
\usepackage{graphicx}
\usepackage{algorithmic}
\usepackage{authblk}
\usepackage{multirow}
\usepackage{amsmath}
\usepackage{amssymb}
\usepackage{amsthm}
\usepackage[T1]{fontenc}
\usepackage{times}
\usepackage{geometry}
\newcommand{\call}{\mathcal{C}_\mathrm{all}}

\newcommand{\ie}{{\em i.e.},~}

\newcommand{\calsple}{\mathcal{S}_{P,\mathit{LE}}}

\newcommand{\tmin}{t_{\mathrm{min}}}

\newcommand{\sinit}{s_{\mathrm{init}}}
\newcommand{\outputs}{\pi_{\mathrm{out}}}
\newcommand{\cinit}[1]{C_{\mathrm{init},#1}}

\newcommand{\sch}{\mathbf{\Gamma}}

\newtheorem{theorem}{Theorem}
\newtheorem{lemma}{Lemma}
\newtheorem{col}{Corollary}
\newtheorem{proposition}{Proposition}
\newtheorem{observation}{Observation}
\theoremstyle{definition}

\newcommand{\ex}{\mathbf{E}}
\newcommand{\var}{\mathbf{Var}}

\newcommand{\cnlogn}{\lfloor c_v n\ln n \rfloor}
\newcommand{\threshold}{\lceil n^{2/3}\rceil}

\title{Leader Election Requires Logarithmic Time in Population Protocols}
\date{}
\author{Yuichi Sudo\thanks{Corresponding Author: y-sudou@ist.osaka-u.ac.jp} }
\author{Toshimitsu Masuzawa\thanks{masuzawa@ist.osaka-u.ac.jp}}

\affil{Graduate School of Information Science and Technology, Osaka University, Japan}

\begin{document}
\maketitle
\begin{abstract}
\normalsize
This paper shows that every leader election protocol
requires logarithmic stabilization time
both in expectation and with high probability
in the population protocol model.
This lower bound holds even if
each agent has knowledge of the exact size of a population
and is allowed to use an arbitrarily large number of agent states. 
This lower bound concludes that
the protocol given in [Sudo et al., SSS 2019] is
time-optimal in expectation.
\end{abstract}

\section{Introduction}
\label{sec:intro}
We consider the \emph{population protocol} (PP) model \cite{original} in
 this paper.
A network called \emph{population} consists of a large number of automata,
called \emph{agents}.
Agents make \emph{interactions}
 (i.e., pairwise communication) with each other
by which they update their states.
Agents are strongly anonymous:
they do not have identifiers
and they cannot distinguish
their neighbors with the same state.
As with the majority of studies on population protocols\cite{original,AG15,AAE+17,AAG18,GS18,GSU18,MST18,DS18,kanjiko,SOK18},
we assume that the network of agents is a complete graph
and that the scheduler
selects an interacting pair of agents at each step
uniformly at random.

In this paper, we focus on the leader election problem, which is one of the most fundamental
and well studied problems in the PP model.
The leader election problem
requires that starting from a specific initial
configuration,
a population reaches a safe configuration in which exactly one leader exists
and the population keeps that unique leader
thereafter.

\subsection{Related Work}
There have been many works which study the leader election problem in the PP model (Tables \ref{tbl:protocols} and \ref{tbl:lower}). 
Angluin et al.~\cite{original} gave the first leader election protocol,
which stabilizes in $O(n)$ parallel time in expectation and uses
only constant space of each agent,
where $n$ is the number of agents
and ``parallel time'' means
the number of steps divided by $n$.
If we stick to constant space,
this linear parallel time is optimal;
Doty and Soloveichik \cite{DS18} showed that
any constant space protocol requires linear parallel time
to elect a unique leader.
Alistarh and Gelashvili \cite{AG15} made a breakthrough
in 2015;
they achieved poly-logarithmic stabilization time
($O(\log^3 n)$ parallel time)
by increasing the number of states from $O(1)$
to only $O(\log^3 n)$.
Thereafter, the stabilization time has been improved
by many studies \cite{BCER17,AAG18,GS18,GSU18,MST18}.
G{\k{a}}sieniec et al.~\cite{GSU18} gave
a state-of-art protocol that stabilizes in
$O(\log n \cdot \log \log n)$ parallel time
with only $O(\log \log n)$ states.
Its space complexity is optimal;
Alistarh et al.~\cite{AAE+17} showed that
any poly-logarithmic parallel time algorithm
requires $\Omega(\log \log n)$ states.
Michail et al.~\cite{MST18} gave a protocol
with $O(\log n)$ parallel time 
but with a linear number of states.
Our previous work \cite{SOT+19} gave a protocol
with $O(\log n)$ parallel time
and $O(\log n)$ states.
Those protocols with non-constant
number of states~\cite{AG15,AAE+17,BCER17,AAG18,GS18,GSU18} are not \emph{uniform}, that is, they require
some rough knowledge of $n$.
For example,
in the protocol of \cite{GS18},
a $\Theta(\log \log n)$ value
must be hard-coded to set the maximum value of
one variable (named $l$ in that paper).
One can find detailed information about the leader election in the PP model in two survey papers \cite{AG18,ER18}.

There is a folklore that
any leader election protocol requires
$\Omega(\log n)$ parallel time in the population protocol model.
One may think that this lower bound trivially holds
because
several agents have no interactions
during $o(\log n)$ parallel time with probability $1-o(1)$.
However, as Alistarh and Gelashvili \cite{AG15}
pointed out, this idea is not sufficient
to prove the folklore.
Let us discuss it in detail here.
The lower bound of $\Omega(\log n)$ expected parallel time
holds almost trivially if
the initial output of the agents is $L$
(\ie all the agents are leaders initially).
This is because 
we need $\Omega(\log n)$ expected parallel time
before $n-1$ agents have at least one interaction each.
What if the initial output is $F$
(\ie all the agents are non-leaders initially)?
For any small constant $\epsilon$,
we can prove that with a constant probability,
$\Omega(n^{1-\epsilon})$ agents remains still \emph{inexperienced}
after the first period of $o(\log n)$ parallel time in an execution, 
that is, they have no interactions during the period.
However, this does not immediately mean that
$\Omega(\log n)$ parallel time is necessary to elect a leader
in expectation
because those $\Omega(n^{1-\epsilon})$ inexperienced agents
are non-leaders.
We have to show that
no leader election protocol 
can create a unique leader with $o(\log n)$ expected parallel time
starting from the initial configuration where all agents are non-leaders.
To the best of our knowledge, 
there is no proof in the literature for this folklore,
that is, the lower bound of $\Omega(\log n)$ parallel time
on the stabilization time for leader election.

\begin{table}
\center
\caption{Leader Election Protocols (Stabilization time is shown in terms of expected parallel time)}
\label{tbl:protocols}
\vspace{0.3cm}

\begin{tabular}{c c c}
\hline
 & States & Stabilization Time\\
\hline
\cite{original}
& $O(1)$ & $O(n)$\\
\cite{AG15}
&$O(\log^3 n)$ & $O(\log^3 n)$\\
\cite{AAE+17}& $O(\log^2 n)$
&$O(\log^{5.3}n \cdot \log \log n)$\\
\cite{AAG18}&
$O(\log n)$&$O(\log^2 n)$\\
\cite{GS18}&
$O(\log \log n)$&$O(\log^2 n)$\\
\cite{GSU18}&
$O(\log \log n)$&$O(\log n \cdot \log \log n)$\\
\cite{MST18}& $O(n)$&$O(\log n)$\\
\cite{SOT+19}&
$O(\log n)$&$O(\log n)$\\
\hline
\end{tabular}
\end{table}

\begin{table}
\center
\caption{Lower Bounds for Leader Election (Stabilization time is shown in terms of expected parallel time)}
\label{tbl:lower}
\vspace{0.3cm}

\begin{tabular}{c c c}
\hline
 & States & Stabilization Time\\
\hline
\cite{DS18} & $O(1)$ & $\Omega(n)$\\
\cite{AAE+17} & $<1/2 \log \log n$ & $\Omega(n/ \mathrm{polylog} n)$\\
This work & any large  & $\Omega(\log n)$\\
\hline
\end{tabular}
\end{table}

\subsection{Our Contribution}
In this paper, we prove the above folklore,
that is, we show that 
any leader election protocol requires
$\Omega(\log n)$ parallel time in expectation.
As mentioned above, most of recent protocols
uses a non-constant (poly-logarithmic, in most cases)
number of states and assume that rough knowledge of the population size
is given to each agent.
This lower bound holds even if
each agent can use an arbitrarily large number of states
and knows the exact size of a population.
Thus, by this lower bound, 
we can say that the protocols of \cite{MST18} and \cite{SOT+19}
are optimal in terms of expected stabilization time.

In our proof for the lower bound,
we do not assume that every leader election protocol always stabilizes to elect a unique leader. 
Therefore, our lower bound holds even if we
allow a protocol to have a (small) probability
that it fails to elect a unique leader.

Strictly speaking, we give a stronger lower bound than
$\Omega(\log n)$ parallel stabilization time
\emph{in expectation}.
Instead, we show that every leader election protocol
requires $\Omega(\log n)$ parallel time to stabilize
with probability $1-o(1)$.
This lower bound immediately gives the above
lower bound in expectation.
Moreover, it immediately yields that
no leader election protocol stabilizes
within $o(\log n)$ parallel time
with high probability;
every leader election protocol stabilizes
within $o(\log n)$ parallel time
with probability $o(1)$.

To prove the lower bound, we introduce a novel notion
that we call $\emph{influencers}$.
At any time of an execution, the influencers of an agent $v$
is the set of agents that could influence on the current state of $v$.
The size of the influencers is monotonically non-decreasing,
and grows with the same speed as \emph{epidemics}, which Angluin et al.~\cite{fast} introduced in order to analyze fast protocols to compute any semi-linear predicate.
Actually, 
we will prove the lower bound essentially by showing that
$\Omega(\log n)$ parallel time is necessary for the number of influencers of any agent $v$
to reach $\Omega(n^{2/3})$.



\section{Preliminaries}
\label{sec:pre}
In this section, we specify the population protocol model.
For simplicity, we omit some elements of the population protocol model
that are not needed to study leader election.
Specifically, we remove input symbols and input functions
from the definition of population protocols.

A \emph{population} is
a network consisting of {\em agents}.
We denote the set of all the agents by $V$ and let $n = |V|$.
We assume that a population is a complete graph,
thus every pair of agents $(u,v)$ can interact,
where $u$ serves as the \emph{initiator}
and $v$ serves as the \emph{responder} of the interaction.

A \emph{protocol} $P(Q,\sinit,T,Y,\outputs)$ consists of 
a finite set $Q$ of agent states,
an initial state $\sinit \in Q$,
a transition function
$T:  Q \times Q \to Q \times Q$,
a finite set $Y$ of output symbols, 
and an output function $\outputs : Q \to Y$.
Every agent is in state $\sinit$
when an execution of protocol $P$ begins.
When two agents interact,
$T$ determines their next states
according to their current states.
The \emph{output} of an agent is determined by $\outputs$:
The output of an agent in state $q$ is $\outputs(q)$.
As with all papers listed in Table \ref{tbl:protocols} except for \cite{original},
we assume that a rough knowledge of $n$ is available.
Specifically, we assume that an integer
$m$ such that $m \ge \log_2 n$
and $m=\Theta(\log n)$ is given,
thus we can design $P(Q,\sinit,T,Y,\outputs)$
using this input $m$,
\ie $Q$, $\sinit$, $T$, $Y$,
and $\outputs$ can depend on $m$.

A \emph{configuration} is a mapping $C : V \to Q$ that specifies
the states of all the agents.
We define $\cinit{P}$ as the configuration of $P$
where every agent is in state $\sinit$.
We say that a configuration $C$ changes to $C'$ by the interaction
$e = (u,v)$,
denoted by $C \stackrel{e}{\to} C'$,
if
$(C'(u),C'(v))=T(C(u),C(v))$
and $C'(w) = C(w)$
for all $w \in V \setminus \{u,v\}$.

A \emph{schedule} $\gamma = \gamma_0,\gamma_1,\dots
=(u_0,v_0),(u_1,v_1),\dots~$
is a sequence of interactions.
A schedule determines which interaction occurs at each \emph{step},
\ie interaction $\gamma_t$ happens at step $t$
under schedule $\gamma$.
We consider a \emph{uniformly random scheduler}
$\sch=\Gamma_0, \Gamma_1,\dots$
where
each $\Gamma_t$ ($t \ge 0$) 
is a random variable that specifies
the interaction $(u_t,v_t)$ at \emph{step} $t$
and satisfies
$\Pr(\Gamma_t = (u,v)) =\frac{1}{n(n-1)}$
for any distinct $u,v \in V$.
Given a schedule $\gamma=\gamma_0,\gamma_1,\dots$,
the \emph{execution} of protocol $P$ starting from a configuration
$C_0$ is uniquely defined as
$\Xi_{P}(C_0,\gamma) = C_0,C_1,\dots$ such that
$C_t \stackrel{\gamma_t}{\to} C_{t+1}$ for all $t \ge 0$. 
We usually focus on $\Xi_{P}(\cinit{P},\sch)$.
We say that agent $v \in V$ \emph{participates} in $\Gamma_t$
if $v$ is either the initiator or the responder of $\Gamma_t$.
We say that a configuration $C$ of protocol $P$ is reachable 
if the initial configuration $\cinit{P}$
changes to $C$ by some finite sequence of interactions
$\gamma_0,\gamma_1,\dots,\gamma_{k}$.
We define $\call(P)$ as the set of all reachable configurations of $P$.

The leader election problem requires that 
every agent should output $L$ or $F$ which means
``leader'' or ``follower'' respectively.
Let $\calsple$ be the set of the configurations of $P$
such that each $C \in \calsple$ satisfies the following:
\begin{itemize}
 \item exactly one agent outputs $L$ (\ie is a leader) in $C$,
and 
 \item no agent changes its output 
in the execution $\Xi_P(C,\gamma)$ for any schedule $\gamma$.
\end{itemize}
We call the configurations of $\calsple$
the \emph{safe} configurations of $P$.
We say that an execution of $P$ \emph{stabilizes}
when it reaches a configuration in $\calsple$.
For any leader election protocol $P$, 
we define the \emph{stabilization time} of $P$
as the number of steps during which
execution $\Xi_P(\cinit{P},\sch)$
reaches a configuration in $\calsple$,
divided by the number of agents $n$.
The division by $n$ implies
that we evaluate the stabilization time
in terms of parallel time.
Since $\sch$ is a random variable,
the stabilization time of $P$ is also a random variable.
Thus, we usually evaluate it in terms of
``in expectation'' or ``with high probability''.


\section{Lower Bound}
Let $P(Q,\sinit,T,Y,\outputs)$
be any leader election protocol.
We fix protocol $P$ and its execution
$\Xi = \Xi_{P}(\cinit{P},\sch) = C_0,C_1,\dots$
throughout this section.
We call $C_t$ the configuration at step $t$
or $t$-th configuration.
Note that each $C_t$ is a random variable.
Our goal is to prove the following proposition.

\begin{proposition}
\label{prop:goal}
For some constant $c$,
the (parallel) stabilization time of $P$
is at least $c\ln n$ with probability $1-o(1)$.
\end{proposition}





We prove Proposition \ref{prop:goal}
in the rest of this section.
First, we prove the following lemma
in a similar way as a standard analysis
of the coupon collector's problem.
\begin{lemma}
\label{lem:sublinear}
Let $\epsilon$ be any (small) positive constant
and $f(n)$ be any function such that
$f(n)=O(n^{1-\epsilon})$.
There exists some constant $c$ such that
execution $\Xi$ requires at least $cn \ln n$ steps
with probability $1-o(1)$
to reach a configuration where less than $f(n)$ agents
are in state $\sinit$.

\end{lemma}
\begin{proof}
Without loss of generality,
we assume that an agent never gets state $\sinit$
once it has an interaction.
(A transition going back to $\sinit$ just increases
the probability that $\Xi$ requires $\Omega(n \log n)$
steps to reach a configuration with less than $f(n)$ agents
in state $\sinit$.)
Consider a configuration that
exactly $i$ agents are in $\sinit$.
Then, at least one of the $i$ agents
has an interaction and leaves state $\sinit$
in the next step with probability
$p_i=\frac{_iC_2 + i(n-i)}{_n C_2}=\frac{i(2n-i-1)}{n(n-1)}$.
Let $X_i$ be a geometric random variable 
with parameter $p_i$,
that is, the number of coin flips until it lands on heads
where the coin lands on heads with probability $p_i$ in each flip.
Let $X=\sum_{i \in \{f^*(n),f^*(n)+2,\dots,n\}}X_i$
where $f^*(n) = 2 \lceil f(n)/2 \rceil$.
Since $p_{n}=p_{n-1}=1$ and $p_i$ is monotonically
increasing in $i\in[0,n-1]$,
for any integer $a$,
the probability that $\Xi$ requires at least $a$ steps
to reach a configuration with less than $f(n)$ agents
in state $\sinit$ is lower bounded by $\Pr(X \ge a)$.
Thus, it suffices to show $\Pr(X \ge c n\log n) = 1-o(1)$
for some constant $c$.

In what follows,
we analyze the expectation and the variance of $X$
and then obtain $\Pr(X \ge c n\log n) = 1-o(1)$
by Chebyshev's inequality.
We obtain the lower bounds of the expectation
and the variance as follows:
\begin{align*}
 \ex[X]&=
\sum_{i \in \{f^*(n),f^*(n)+2,\dots,n\}} \frac{1}{p_i}
\ge \sum_{i \in \{f^*(n),f^*(n)+2,\dots,n\}} \frac{n}{2i}
=
\Omega\left(n \log \frac{n}{f(n)}\right) = \Omega (n\log n^{\epsilon}) = \Omega(n \log n),\\
\var[X]&=\sum_{i \in \{f^*(n),f^*(n)+2,\dots,n\}} \frac{1-p_i}{p_i^2}\le \sum_{i=1,2,\dots,n}\frac{1}{p_i^2}
\le \sum_{i=1,2,\dots,\infty}\frac{n^2}{i^2}= \frac{\pi^2 n^2}{6}
< 2n^2,
\end{align*}
where we use $1^2+2^2+3^2+\dots = \pi^2/6$ for the last equality.
Let $d$ be a constant such that $\ex[X] \ge d n\ln n$ holds
for any sufficiently large $n$.
Then, by Chebyshev's Inequality, we obtain
\begin{align*}
\Pr\left(X \le \frac{dn \ln n}{2}\right) \le \Pr\left(X \le E[X]-\frac{dn \log n}{2}\right)
\le \frac{4\var[X]}{(dn \ln n)^2}=O(1/\log^2 n).
\end{align*}
Thus, we have $\Pr(X > dn \ln n /2)=1-O(1/\log^2 n)=1-o(1)$.
\end{proof}

\begin{col}
\label{col:initialleader}
Proposition \ref{prop:goal} holds
if the initial output of $P$ is $L$,
\ie $\outputs(\sinit)=L$.
\end{col}

In the rest of this section,
we assume $\outputs(\sinit)=F$.
Recall that a configuration $C$ of $P$ is safe
if and only if
there exists exactly one leader in $C$
and no agent can change its output in
an execution after $C$.
In what follows, we use Lemma \ref{lem:sublinear}
by letting $f(n)=n^{2/3}$ while 
the lemmas and corollaries in the rest of this section
hold for more general $f(n)=O(n^{1-\epsilon})$.
The following corollary immediately follows from Lemma \ref{lem:sublinear}.

\begin{col}
\label{col:initlimit}
Suppose that Proposition \ref{prop:goal} does not hold,
that is, the parallel stabilization time of $P$
is less than $c \ln n$
with probability $1-o(1)$
for any constant $c$.
Then, there exists some safe configuration of $P$
where at least $n^{2/3}$ agents are in state $\sinit$.
\end{col}

\begin{figure} 
\centering
\begin{minipage}{0.45\hsize}
\centering
\includegraphics[width= 0.9\linewidth,clip]{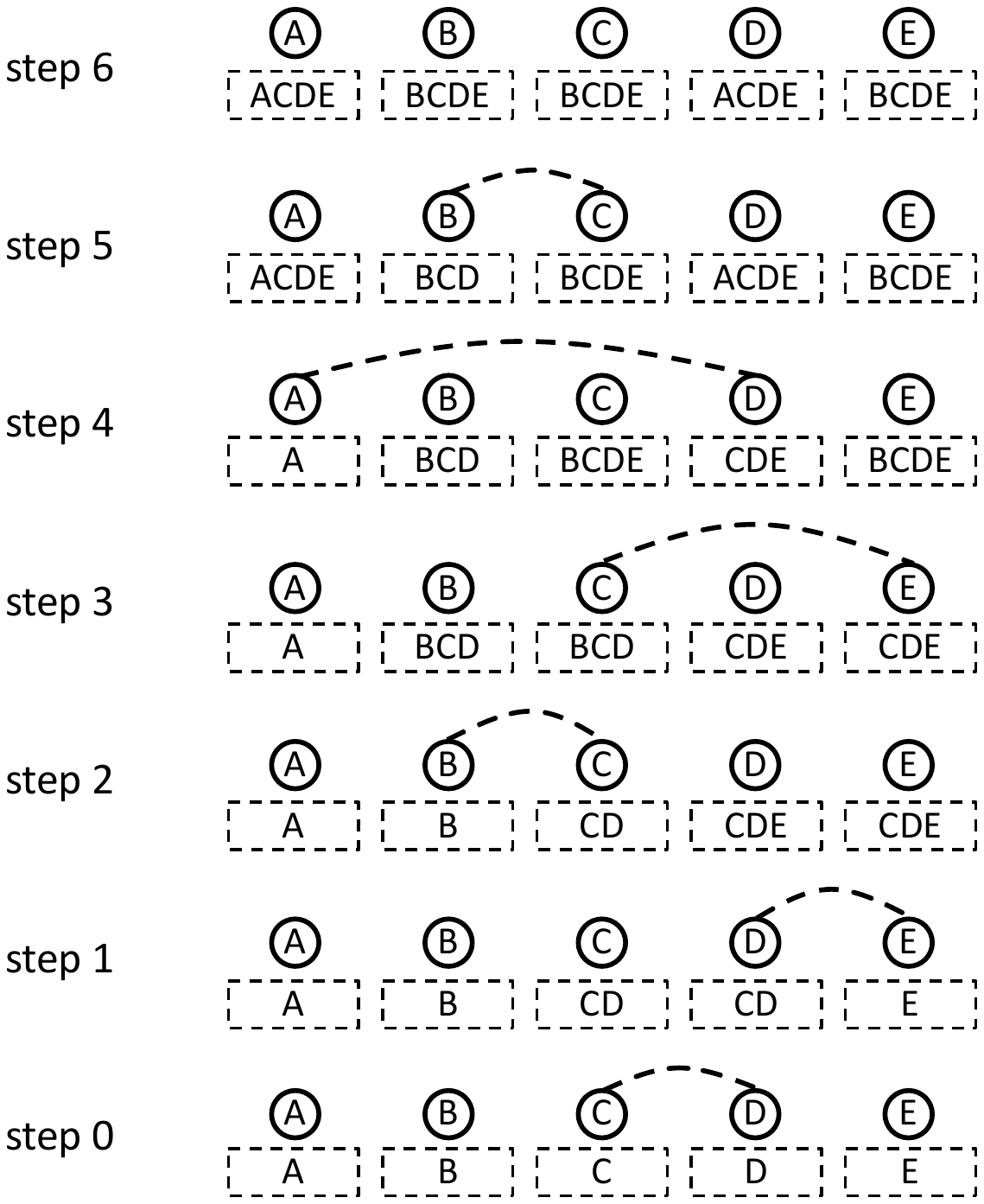} 
\caption{An example of the influencers. The circles represent the agents.
The dashed lines represents the interactions in steps 0, 1, \dots, 5.
The box below each circle
represents the set of influencers of the corresponding agent
at each step.
} 
\label{fig:influencer} 
\end{minipage}
\begin{minipage}{0.03 \hsize}
\hspace{0.01cm}
\end{minipage}
\begin{minipage}{0.45\hsize}
\centering
\vspace{-0.85cm}
\includegraphics[width= 0.68\linewidth,clip]{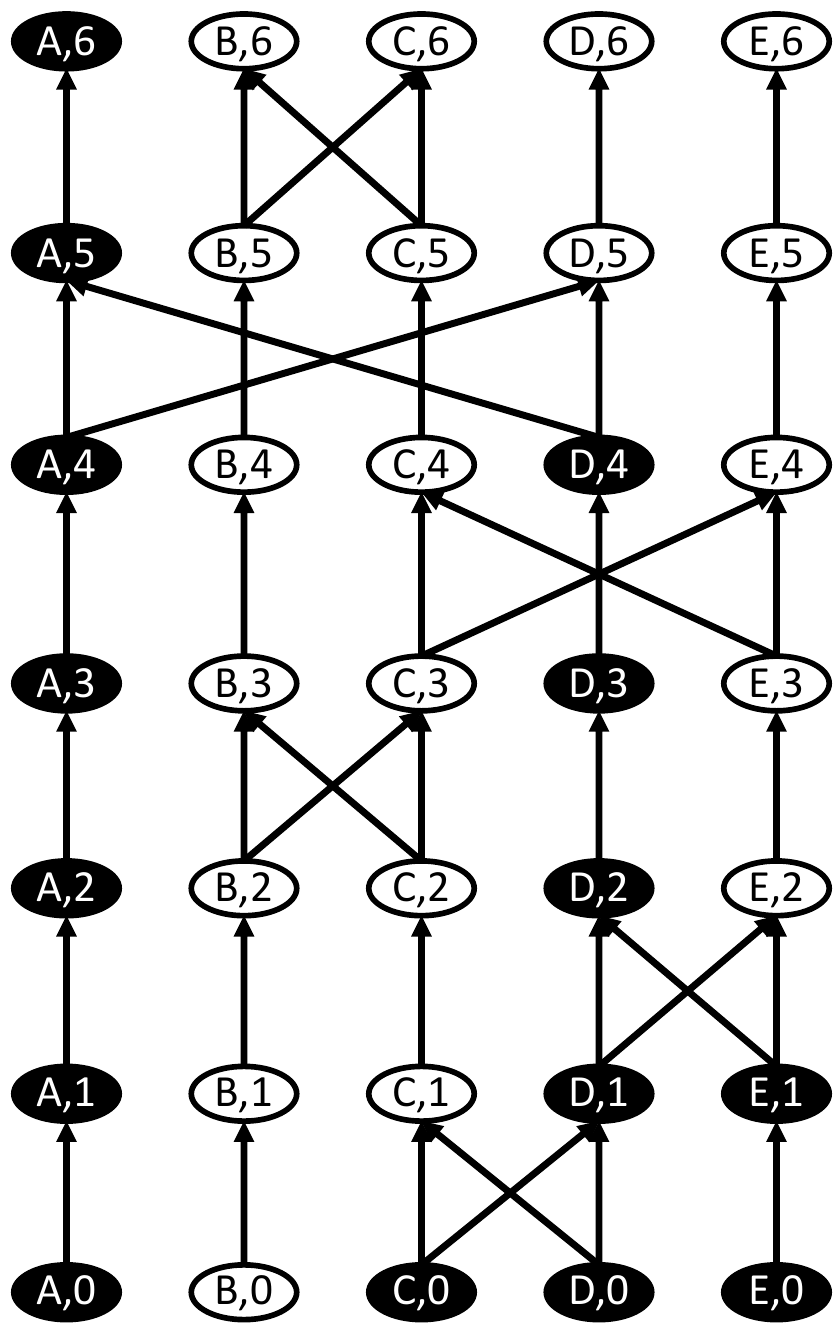} 
\caption{The graph $H$ that corresponds to the interaction sequence in Figure \ref{fig:influencer}. The black ellipses represent the nodes from which $(A,6)$ is reachable.} 
\label{fig:graph_H}  
\end{minipage}
\end{figure} 

Corollary \ref{col:initlimit} implies that
an execution of $P$ must involve more than $n^{2/3}$ agents
to create a new leader
if Proposition \ref{prop:goal} does not hold.
This is because otherwise
an execution of $P$ creates a new leader
with only interactions involving only at most $n^{2/3}$ agents,
a contradiction to the existence of a safe configuration
with at least $n^{2/3}$ agents in state $\sinit$.
In what follows, we elaborate this proposition
as Lemma \ref{lem:influencer} after introducing the notion of
\emph{influencer}.
The set of influencers of agent $v$ at step $0$, denoted
by $F(v,0)$, is only $\{v\}$.
Thereafter, the influencers of agent $v$
is expanded every time it has an interaction with another agent.
Specifically,  
for $i > 0$,
$F(v,i)=F(v,i-1) \cup F(u,i-1)$ if
$v$ has an interaction with an agent $u$ at step $i$,
that is, if there exists
agent $u \in V$ such that 
$\Gamma_{i-1} =(u,v)$ or $\Gamma_{i-1}=(v,u)$.
Otherwise, $F(v,i)=F(v,i-1)$.
See Figure \ref{fig:influencer} that depicts the set of influencers where the population consists of five agents $\{A,B,C,D,E\}$.
In this example, by the interactions at steps 0, 1, \dots, 5, the set of the influencers of agent $A$
expands from $\{A\}$ to $\{A,C,D,E\}$.
We can represent $F(v,t)$ more intuitively. 
Consider the directed graph $H=(V_H,E_H)$
where $V_H=\{(u,i) \mid u \in V, i=0,1,\dots,t\}$
and $E_H$ is defined as follows:
\begin{align*}
E_H=&\{((u,i),(u,i+1))\} \mid u \in V, i=0,1,\dots,t-1\}\\
&\cup \{((u,i),(w,i+1))\} \mid u,w \in V, i=0,1,\dots,t-1, (u,w) \in \Gamma_i \vee (w,u) \in \Gamma_i\}. 
\end{align*}
(See Figure \ref{fig:graph_H} for
the graph $H$ that corresponds to the example
of Figure \ref{fig:influencer}.)
It is obvious that
a node $u$ belongs to $F(v,t)$
if and only if
node $(v,t)$ is reachable from node $(u,0)$
in graph $H$.

\begin{lemma}
\label{lem:influencer}
If
Proposition \ref{prop:goal} does not hold,
an execution of $P$ never reaches a safe configuration
before the number of influencers of some agent becomes
greater than $n^{2/3}$,
that is, $C_t$ is a safe configuration
only if $|F(v,t)| > n^{2/3}$ holds for some $v \in V$.
\end{lemma}
\begin{proof}
Assume that Proposition \ref{prop:goal} does not hold.
Then, by Corollary \ref{col:initlimit},
there exists a safe configuration $C$ of $P$ such that
$m \ge n^{2/3}$ agents are in state $\sinit$.
Since $C$ is a safe configuration,
there is no sequence of interactions that leads to create
another leader starting from $C$.
This means that
we cannot create a new leader by interacting only $m$ agents
with state $\sinit$ even if we let them interact each other
infinitely many times.
Therefore, we require that the number of influencers
of some agent becomes greater than $m \ge n^{2/3}$
to create a new leader.
In other words, in execution $\Xi$,
an agent $v$ becomes a leader only at step $t$
such that $|F(v,t)| > n^{2/3}$.
The lemma holds because no leader exists in a configuration $C_0=\cinit{P}$
and thus $\Xi$ must create a leader to reach a safe configuration.
\end{proof}

By Lemma \ref{lem:influencer},
it suffices to show that
the expansion of influencers is not so fast
in order 
to prove Proposition \ref{prop:goal}.
More specifically, our goal is now to show that
$\Omega(n\log n)$ steps are needed until
some agent $v \in V$ satisfies $F(v,*) > n^{2/3}$.
Fortunately, the expansion of influencers
is symmetric to the expansion of the epidemic \cite{fast}
and can be analyzed similarly.
Let $t$ be any non negative integer.
We define a sequence of sets
$I_{v,t}(0), I_{v,t}(1), \dots, I_{v,t}(t) \in 2^V$
based on the digraph $H$
defined just above Lemma \ref{lem:influencer},
as follows:
for any $i = 0,1,\dots,t$,
a node $u \in V$ belongs to $I_{v,t}(i)$
if and only if $(v,t)$ is reachable from $(u,i)$ in $H$.
In the example of Figure \ref{fig:graph_H},
we have $I_{A,6}(6)=I_{A,6}(5)=\{A\}, I_{A,6}(4)= I_{A,6}(3)= I_{A,6}(2)=\{A,D\}, I_{A,6}(1)=\{A,C,D\}, I_{A,6}(0)=\{A,C,D,E\}$.
By definition, we have the following observation.
\begin{observation}
\label{observ:fandi}
Let $v \in V$ and $t \in \mathbb{N}_{\ge 0}$.
Then, we have $F(v,t)=I_{v,t}(0)$. 
\end{observation}

Let $i \in [0,t-1]$.
Note that
$I_{v,t}(i)$ is determined only by
interactions $\Gamma_{t-1}, \Gamma_{t-2},\dots,\Gamma_{i}$.
Hence, $I_{v,t}(i)$ depends on $I_{v,t}(i+1)$,
but $I_{v,t}(i+1)$ is independent of $I_{v,t}(i)$.
Suppose $|I_{v,t}(i+1)|=k$.
Then, $|I_{v,t}(i)|=k+1$ holds
if and only if one of the $k$ agents in $I_{v,t}(i+1)$
and one of the $n-k$ agents in $V \setminus I_{v,t}(i+1)$
interact at step $i$ (\ie in $\Gamma_i$).
Therefore, we have the following observation.
\begin{observation}
\label{observ:prob}
Let $v \in V$ and $t \in \mathbb{N}_{\ge 0}$.
Then, we have $0 \le |I_{v,t}(i)|-|I_{v,t}(i+1)| \le 1$
and $\Pr(|I_{v,t}(i)|=k+1 \mid |I_{v,t}(i+1)|=k) =\frac{2k(n-k)}{n(n-1)}$
for any integer $k=1,2,\dots,n$.
\end{observation}

We show
that the above sufficient condition for Proposition \ref{prop:goal} holds, as the following lemma.

\begin{lemma}
\label{lem:growspeed}
Let $\tmin$ be the smallest integer
such that $|F(v,\tmin)| > n^{2/3}$ holds for some $v \in V$.
Then, there exists some constant $c$ such that
$\Pr(\tmin \ge cn \ln n) = 1-o(1)$.
\end{lemma}
\begin{proof}
Let $v$ be any agent in $V$.
In what follows,
we show $\Pr(|F(v,\cnlogn)|> \threshold)=O(n^{-2})$
for some constant $c_v$.
This yields $\Pr(\tmin < cn\ln n) = O(n^{-1})=o(1)$
by the union bounds
where $c= \min \{c_v \mid v \in V\}$.

By Observation \ref{observ:fandi},
$F(v,\cnlogn)=I_{v,\cnlogn}(0)$ holds.
Therefore,
letting $X_k$ be a geometric random variable
with parameter $p_k= \frac{2k(n-k)}{n(n-1)}$
and $S_{i,j}=\sum_{i \le k \le j} X_k$, 
we obtain the following inequality
by Observation \ref{observ:prob}:
\begin{align*}
\Pr(|F(v,\cnlogn)| > \threshold)=\Pr(|I_{v,\cnlogn}(0)|>\threshold)=\Pr(S_{1,\threshold} \le \cnlogn).
\end{align*}
Let $r = \lfloor \sqrt{n} \rfloor$ and
$\kappa = \lfloor \threshold / r \rfloor$.
To make use of Chernoff bounds, 
we divide $S_{1,\threshold}$ to $\kappa=O(n^{1/6})$ groups,
$S_{1,r},S_{r+1,2r},\dots,S_{(\kappa-1) r+1,\kappa r}$.
\footnote{
We ignore the last segment $S_{\kappa r+1,\threshold}$
when $\threshold$ is not divisible by $r$.
This ignorance only increase the error probability
and thus does not ruin the proof, as we will see in the last
sentence of the proof of this lemma.
}
Rename $S'_i=S_{ir+1,(i+1)r-1}$ for any $i=0,1,\dots,\kappa-1$.
While $k\le n/2$, probability $p_k$ is monotonically increasing.
Since $\threshold \ll n/2$ holds for sufficiently large $n$,
we can assume $\ex[X_1]>\ex[X_2]> \cdots >\ex[X_{\threshold}]$.
Thus, letting $B(l,p)$ be a binomial random variable
with parameters $l$ and $p$,
where $l$ is the number of trials and $p$ is the success probability, we have
\begin{align*}
\Pr\left (S'_i \le \left\lfloor\frac{r}{2}\cdot \ex[X_{(i+1)r}]
\right \rfloor \right)
&<
\Pr\left (B\left(\left\lfloor\frac{r}{2}\cdot \ex[X_{(i+1)r}]\right \rfloor,p_{(i+1)r}\right)\ge r\right)\\
&\le \exp\left(-\frac{1}{3} \cdot \left(\frac{r}{2}-1\right)\right)\\
&\ll n^{-3}
\end{align*}
for sufficiently large $n$,
where we use the Chernoff Bound for the second inequality.
Let $E'=\sum_{0 \le i < \kappa}\left\lfloor\frac{r}{2}\ex[X_{(i+1)r}]\right \rfloor$.
Then, we have
\begin{align*}
E'
= \sum_{0 \le i < \kappa}
\left \lfloor
\frac{r}{2}\cdot
\frac{n(n-1)}{2(i+1)r(n-(i+1)r)}\right \rfloor
= \Omega\left(
\sum_{0 \le i < \kappa}
\frac{n}{i}
\right)
= \Omega(n \log \kappa) = \Omega(n \log n).
\end{align*}
Thus, for some (small) constant $c_v$
and sufficiently large $n$, 
we have $\cnlogn < E'$.
To conclude, we have
\begin{align*}
\Pr(|F(v,\cnlogn)|> \threshold) &= \Pr(S_{1,\threshold} \le \cnlogn)\\
&<
\Pr\left (\sum_{0 \le i < \kappa}S'_i \le \cnlogn \right)\\
&< \Pr\left (\sum_{0 \le i < \kappa}S'_i \le E'\right)\\
&< \sum_{0 \le i < \kappa} \Pr\left(S'_i \le \left\lfloor\frac{r}{2}\ex[X_{(i+1)r}]\right \rfloor\right)\\
&\ll n^{-2}
\end{align*}
for sufficiently large $n$.
\end{proof}

\begin{theorem}
\label{theorem:result} 
Proposition \ref{prop:goal} holds.
That is, 
every leader election protocol requires 
$\Omega(\log n)$ (parallel) stabilization time
with probability $1-o(1)$.
\end{theorem}
\begin{proof}
Assume that the expected parallel stabilization time
of protocol $P$ is $o(\log n)$.
By Lemma \ref{lem:influencer},
an execution cannot reach a safe configuration
before $|F(v,t)| \ge n^{2/3}$ holds for some $v \in V$.
However, Lemma $\ref{lem:growspeed}$
yields that this requires $\Omega(\log n)$ parallel time,
contradiction.
\end{proof}
The following two theorems immediately follows from
Theorem \ref{theorem:result}.
\begin{theorem}
Every leader election protocol requires $\Omega(\log n)$
(parallel) stabilization time in expectation.
\end{theorem}
\begin{theorem}
No leader election protocol
stabilizes within $o(\log n)$ time
with high probability (\ie with probability $1-O(n^{-1})$).
\end{theorem}


\section{Conclusion}
In this paper, we proved that
in the population protocol model,
any leader election protocol requires
$\Omega(\log n)$ parallel stabilization time
both in expectation and with high probability.
This lower bound holds even if
the protocol use an arbitrarily large number of agent states
and each agent knows the exact size $n$ of a population.


\section*{Acknowledgments}
This work was supported by JSPS KAKENHI Grant Numbers
17K19977, 18K18000, and 19H04085
and JST SICORP Grant Number JPMJSC1606.
\bibliographystyle{unsrt}
\bibliography{paper3} 

\begin{thebibliography}{10}

\bibitem{original}
Dana. Angluin, James Aspnes, Zo\"{e} Diamadi, Michael~J. Fischer, and Ren\'{e}
  Peralta.
\newblock Computation in networks of passively mobile finite-state sensors.
\newblock {\em Distributed Computing}, 18(4):235--253, 2006.

\bibitem{AG15}
Dan Alistarh and Rati Gelashvili.
\newblock Polylogarithmic-time leader election in population protocols.
\newblock In {\em Proceedings of the 42nd International Colloquium on Automata,
  Languages, and Programming}, pages 479--491. Springer, 2015.

\bibitem{AAE+17}
Dan Alistarh, James Aspnes, David Eisenstat, Rati Gelashvili, and Ronald~L
  Rivest.
\newblock Time-space trade-offs in population protocols.
\newblock In {\em Proceedings of the Twenty-Eighth Annual ACM-SIAM Symposium on
  Discrete Algorithms}, pages 2560--2579. SIAM, 2017.

\bibitem{AAG18}
Dan Alistarh, James Aspnes, and Rati Gelashvili.
\newblock Space-optimal majority in population protocols.
\newblock In {\em Proceedings of the Twenty-Ninth Annual ACM-SIAM Symposium on
  Discrete Algorithms}, pages 2221--2239. SIAM, 2018.

\bibitem{GS18}
Leszek G{\k{a}}sieniec and Grzegorz Stachowiak.
\newblock Fast space optimal leader election in population protocols.
\newblock In {\em Proceedings of the Twenty-Ninth Annual ACM-SIAM Symposium on
  Discrete Algorithms}, pages 2653--2667. SIAM, 2018.

\bibitem{GSU18}
Leszek G{\k{a}}sieniec, Grzegorz Stachowiak, and Przemyslaw Uznanski.
\newblock Almost logarithmic-time space optimal leader election in population
  protocols.
\newblock In {\em The 31st ACM on Symposium on Parallelism in Algorithms and
  Architectures}, pages 93--102. ACM, 2019.

\bibitem{MST18}
Othon Michail, Paul~G Spirakis, and Michail Theofilatos.
\newblock Simple and fast approximate counting and leader election in
  populations.
\newblock In {\em Proceedings of the 20th International Symposium on
  Stabilizing, Safety, and Security of Distributed Systems}, pages 154--169.
  Springer, 2018.

\bibitem{DS18}
David Doty and David Soloveichik.
\newblock Stable leader election in population protocols requires linear time.
\newblock {\em Distributed Computing}, 31(4):257--271, 2018.

\bibitem{kanjiko}
Yuichi Sudo, Junya Nakamura, Yukiko Yamauchi, Fukuhito Ooshita, Hirotsugu.
  Kakugawa, and Toshimitsu Masuzawa.
\newblock Loosely-stabilizing leader election in a population protocol model.
\newblock {\em Theoretical Computer Science}, 444:100--112, 2012.

\bibitem{SOK18}
Yuichi Sudo, Fukuhito Ooshita, Hirotsugu Kakugawa, Toshimitsu Masuzawa, Ajoy~K
  Datta, and Lawrence~L Larmore.
\newblock Loosely-stabilizing leader election with polylogarithmic convergence
  time.
\newblock In {\em 22nd International Conference on Principles of Distributed
  Systems (OPODIS 2018)}, pages 30:1--30:16, 2018.

\bibitem{BCER17}
Andreas Bilke, Colin Cooper, Robert Els\"{a}sser, and Tomasz Radzik.
\newblock Brief announcement: Population protocols for leader election and
  exact majority with $o(log^2 n)$ states and $o(log^2 n)$ convergence time.
\newblock In {\em Proceedings of the 38th ACM Symposium on Principles of
  Distributed Computing}, pages 451--453. Springer, 2017.

\bibitem{SOT+19}
Yuichi Sudo, Fukuhito Ooshita, Taisuke Izumi, Hirotsugu Kakugawa, and
  Toshimitsu Masuzawa.
\newblock Logarithmic expected-time leader election in population protocol
  model.
\newblock In {\em Proceedings of the 21st International Symposium on
  Stabilization, Safety, and Security of Distributed Systems}, page (to
  appear), 2019.

\bibitem{AG18}
Dan Alistarh and Rati Gelashvili.
\newblock Recent algorithmic advances in population protocols.
\newblock {\em ACM SIGACT News}, 49(3):63--73, 2018.

\bibitem{ER18}
Robert Els\"asser and Tomasz Radzik.
\newblock Recent results in population protocols for exact majority and leader
  election.
\newblock {\em Bulletin of EATCS}, 3(126), 2018.

\bibitem{fast}
Dana Angluin, James Aspnes, and David Eisenstat.
\newblock Fast computation by population protocols with a leader.
\newblock {\em Distributed Computing}, 21(3):183--199, 2008.

\end{thebibliography}
\end{document}